\documentclass{article}

\usepackage{calc}
\usepackage{amsmath}
\usepackage{amsxtra}
\usepackage{amssymb}
\usepackage{amsfonts}
\usepackage{amsthm}
\usepackage{amstext}
\usepackage{amsbsy}
\usepackage{amscd}
\usepackage{color}

\usepackage{graphicx}

\theoremstyle{plain}
\newtheorem{theorem}{Theorem}
\newtheorem{lemma}{Lemma}
\newtheorem{corollary}[lemma]{Corollary}
\newtheorem{proposition}[lemma]{Proposition}
\theoremstyle{definition}

\newtheorem{definition}[lemma]{Definition}
\theoremstyle{remark}

\newcommand{\defm}[1]{\emph{#1}}

\newcommand{\trace}{\operatorname{tr}}

\newcommand{\myeqref}[1]{(\ref{#1})}
\newcommand{\myref}[1]{\ref{#1}}

\newcommand{\mylabel}[1]{\label{#1}}
\newcommand{\myeqlabel}[1]{\label{#1}}  %\addtocounter{myeqn}{1}

\newcommand{\vv}{\vec v}

\newcommand{\myquote}[1]{%
    \par\hspace{1cm}\parbox{\linewidth-2cm}{#1}\hspace{1cm}\par
}

\begin{document}

\title{Instability of strong regular reflection and counterexamples to the detachment criterion}
\author{Volker Elling}
\date{}	

\maketitle

\begin{abstract}
    We consider a particular instance of reflection of shock waves in self-similar compressible flow.
    We prove that local self-similar regular reflection (RR) cannot always be extended into a global flow. 
    Therefore the detachment criterion is not universally correct.
    More precisely, consider the following \defm{angle condition}:
    the tangent of the strong-type reflected shock meets the opposite wall at a sharp or right downstream side angle.
    In cases where the condition is violated and the weak-type reflected shock is transonic, we show that global RR does not exist.
    Combined with earlier work we have shown that none of
    the classical criteria for RR$\rightarrow$MR transition is universally correct.
    A new criterion is proposed.
    Moreover, we have shown that strong-type RR is unstable,
    in the sense that global RR cannot persist under perturbations to one side.
    This yields a definite answer to the weak-strong problem because earlier work 
    shows \emph{stability} of weak RR in the same sense. 
\end{abstract}

76H05; 76L05

\section{Introduction}

Consider compressible flow. 
In \defm{regular reflection} (RR; see Figure \myref{fig:locrr-left})
an \defm{incident} shock wave meets a solid wall in a \defm{reflection point}
and continues as a second, \defm{reflected} shock.
\begin{figure}[h]
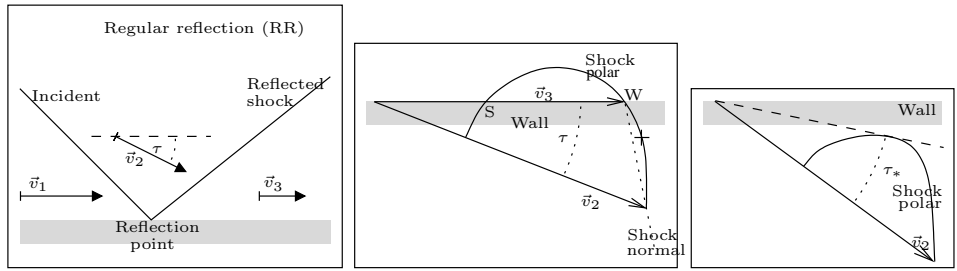

\parbox{1.2\linewidth}{\input{locrr.pstex_t}
\input{polar.pstex_t}
\input{polar2.pstex_t}}
\caption{Left: local regular reflection (RR). Center: shock polar ($\tau<\tau_*$, local RR possible). Right: $\tau>\tau_*$, local RR impossible).}
\mylabel{fig:locrr-left}
\mylabel{fig:polar-center}
\mylabel{fig:polar2-right}
\end{figure}
In many circumstances the effects of viscosity, heat conduction, boundary roughness etc.\ are negligible\footnote{for experimental examples see \cite[p.\ 142f]{van-dyke}}, 
so that inviscid models are appropriate; in this article we focus on compressible polytropic potential flow. 
Then shocks are sharp jumps satisfying Rankine-Hugoniot conditions and the \defm{slip condition} is used at walls: 
$\vec v_1,\vec v_3$ must be tangential. 

Consider a fixed constant state (velocity, density and sound speed) in the $2$-sector and vary the angle of the reflected shock. 
Each angle yields a different $3$-sector state. The curve of possible $\vec v_3$ is 
called \defm{shock polar}. The maximum angle between $\vec v_2$ and $\vec v_3$ 
is called \defm{critical angle}. If it is larger than $\tau$  (Figure \myref{fig:polar-center} center), the angle between $\vec v_2$ and wall, then there are 
\emph{two} possible reflected shocks satisfying the slip condition,
called \defm{weak-type} (W) and \defm{strong-type}\footnote{The names refer to their relative strength, but the absolute strength can be arbitrarily small or large.} (S).

There is no \emph{local} argument to rule out one type; 
the Rankine-Hugoniot and slip conditions allow both. 
At least in initial-value problems we expect uniqueness, in nature and in good mathematical models. 
For this we need to consider the \emph{global} flow that contains the reflection, 
in particular domain shape and far-field/boundary conditions far from the reflection point. 
Of course there is an infinite\footnote{In fact almost all flows with shocks include some form of shock reflection.} variety of such flows, 
but some observations and arguments apply to most if not all of them. 

If $\tau>\tau_*$, on the other hand, then even locally RR is theoretically impossible because none of the 
reflected shock angles can make $\vec v_3$ parallel to the wall (Figure \ref{fig:polar2-right} right).
Around 1875, Ernst Mach \cite{mach-wosyka} 
discovered another pattern, now named \defm{Mach reflection}
(MR; see Figure \myref{fig:highertheta-left} left), where incident and reflected shock meet off the wall in a \defm{triple point} with a third shock, the \defm{Mach stem}. 
For some parameters both RR and MR are possible. 
Starting with John von Neumann \cite{neumann-1943}, many researchers have tried to predict
the precise parameters at which the RR$\rightarrow$MR transition takes place
(see \cite{ben-dor-book,ben-dor-shockwaves2006} for a survey of this and other problems in shock
reflection).

There are three classical transition criteria. The \defm{von Neumann criterion} does not apply in potential flow\footnote{Even in Euler flow 
it applies only for sufficiently high Mach numbers.} at all. 
The \defm{detachment criterion} predicts global RR whenever a local RR exists. 
The \defm{sonic criterion}, in contrast, predicts global RR if and \emph{only} if there is a local RR 
with \defm{supersonic}\footnote{which is necessarily weak-type} reflected shock. 
All three criteria are motivated by local considerations and well-defined for any global problem; 
of course the same criterion need not be correct for all global problems. 
However, we make a stronger observation: in a particular global problem, \emph{none} of the classical criteria is correct, 
so that an entirely new criterion must be found. (The most promising candidates are modifications of the detachment 
criterion.) 

To define our problem we add a second solid wall that meets the original wall
right of the reflection point (Figure \ref{fig:mod} left), enclosing a corner angle $180^\circ-\theta$. 
To satisfy the slip condition in the constant-state $2$-sector, the opposite wall has to move with horizontal speed $\vec w=\vec w(\theta)$ so that
$\vec v_2\cdot\vec n=\vec w\cdot\vec n$ ($\vec n$ wall normal). 

\begin{figure}[h]
\input{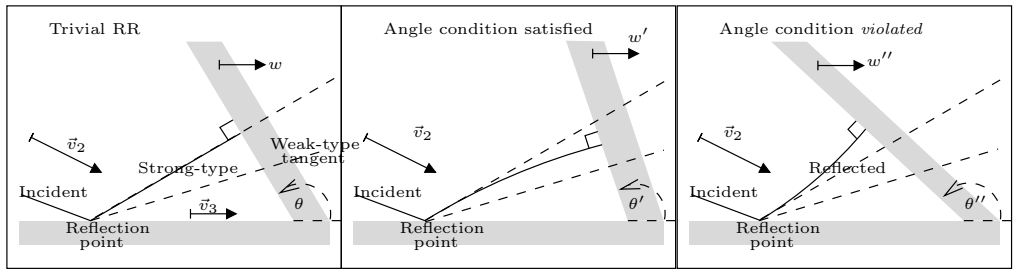}
\caption{Left: trivial RR. Center: angle condition satisfied; may or may not exist. Right: angle condition violated;
no such flow can exist, even if the reflection-point tangent is allowed to be weak-type.}
\label{fig:mod}
\end{figure}

There is exactly one $\theta$ so that the new wall is \defm{perpendicular} to the reflection-point tangent
of the strong-type reflected shock
(Figure \ref{fig:mod} left). 
In this case, $\vec v_1\cdot\vec n=\vec w\cdot\vec n$ as well, so the fluid in the $1$-sector is
\emph{also} compatible with the wall. The result is what we call a \defm{trivial RR}.

However, for any other $\theta$ the reflected shock would have to be curved (and border a non-constant region on its right), because its reflection point tangent does not
form a right angle with the new wall. So there is a large variety of nontrivial cases;
each has the same incident and reflected shock, but $\theta$ and $\vec w$ vary.

Alternatively, we may consider the coordinate system of an observer travelling 
in the wall-wall corner. He observes steady walls but moving shocks (Galilean invariance). 
Moreover, use reflection so that the new (opposite) and old (reflection) wall change places 
(Figure \ref{fig:ini-left}). 

Let $\alpha$ be the counterclockwise angle from incident shock to opposite wall in Figure \myref{fig:ini-left} left.
We have a family of problems, with parameter space consisting of triples $(M_1,\alpha,\theta)$.
At time $t=0$ the incident shock starts in the wall-wall corner (Figure \myref{fig:ini-left} left). 

\begin{figure}[h]
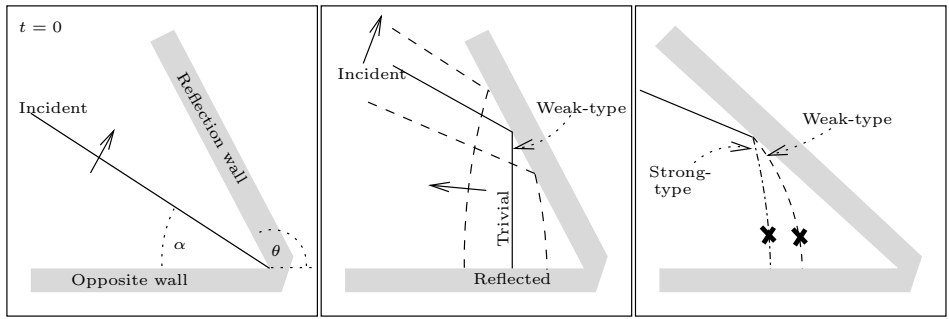

\hbox{\input{ini.pstex_t}\input{pert.pstex_t}\input{nopert.pstex_t}}
\caption{Left: initial data. Center: Weak-type trivial RR and possible perturbations (dashed).
Right: angle condition violated, neither weak-type nor strong-type global RR exist.}
\mylabel{fig:ini-left}
\mylabel{fig:pert-center}
\mylabel{fig:nopert-right}
\end{figure}

\begin{figure}[h]
\input{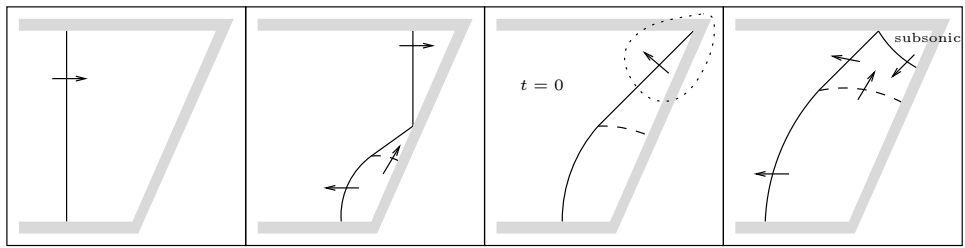}
\caption{Left: $\alpha=90^\circ$ incident approaching; left center: classical RR; right center: upper corner
is locally like Figure \myref{fig:ini-left}.}
\mylabel{fig:experiment}
\end{figure}

Such reflections occur in practice (Figure \myref{fig:experiment}). 
Experimentally a (nearly) straight vertical shock could be produced by breaking diaphragms or detonating small charges.
This shock (Figure \ref{fig:experiment} left) travels to the right through a tube, meeting the lower corner at some time. 
A first reflection occurs (Figure \myref{fig:experiment} left center). 
It is the classical case $\alpha=0^\circ$, $\theta<90^\circ$ which has been studied extensively
\cite{chen-feldman-selfsim-journal,elling-rrefl,yuxi-zheng-rref,canic-keyfitz-kim}. The reflected shock
travels up the wall, reaching a second corner at $t=0$. In that instant, 
the \emph{local} flow near the upper corner is the same
as the initial data in Figure \myref{fig:ini-left} left.

\cite{elling-sonic-potf} has already obtained global \emph{weak}-type transonic RR for small perturbations of the trivial $\theta$,
in the following class:
\begin{definition}
    \mylabel{def:flow-class}%
    Consider self-similar potential flow (see Section \myref{section:potf}).
    A \defm{transonic} (or sonic) global RR (see Figure \myref{fig:trans-left} left) 
    has a straight incident shock extending to infinity,
    meeting the reflected shock in a single \defm{reflection point} on the reflection wall.
    The incident shock separates the $1$- and $2$-sector, two regions of constant fluid state
    $\rho,c,\vv$. The reflected shock is $C^1$ including the endpoints, 
    separating the $2$- from the $3$-sector, meeting the opposite wall in a right angle.
    Flow in the interior of the $3$-sector is elliptic (pseudo-Mach number $L<1$, see 
    \myeqref{eq:pseudo-Mach}), with continuous fluid variables.
\end{definition}
The sonic criterion, in any reasonable precise formulation, predicts non-existence (and appearance of MR), so \cite{elling-sonic-potf} demonstrates
that it cannot be universally correct. 
The present paper considers the case of strong-type RR. 
\begin{definition}
    \mylabel{def:angle-condition}%
    We say $\theta$ satisfies the \defm{angle condition} if the reflection-point tangent 
    of the \emph{strong}-type reflected shock forms
    an angle $\leq 90^\circ$ (Figure \ref{fig:mod} left and center) on its \emph{downstream} side
    with the opposite wall.
\end{definition}

\begin{theorem}
    \mylabel{th:detach-wrong}%
    Consider parameters $M_1,\alpha,\theta$ so that local weak-type transonic RR exists, but 
    the angle condition (Definition \ref{def:angle-condition}) is violated. Then global RR solutions
    of the kind in Definition \myref{def:flow-class} do not exist (\emph{neither} weak-type nor strong-type).
\end{theorem}

\begin{corollary}
    The following version of the detachment criterion is not universally correct:
    \addtocounter{footnote}{1}
    \footnotetext{It is violated for a subset of the parameter space
      which is open and nonempty, hence ``generic'' by any reasonable definition.}
    \newcounter{fndecay}\setcounter{fndecay}{\value{footnote}}%
    \myquote{%
        Generically\footnotemark[\value{fndecay}], when local RR exists, 
        either weak- or strong-type can be extended into a global RR.
    }
\end{corollary}
\begin{proof}
    It is sufficient to give a rigorous proof of existence of local RR
    satisfying the conditions of Theorem \myref{th:detach-wrong}:
    take $\vec x/t=0$ to be the reflection point.
    Choose some supersonic $1$-sector state.
    Then for sufficiently small $\tau$ (Figure \myref{fig:locrr-left} left)
    we can find a weak-type incident shock and a $2$-sector state
    with $M_2>1$, as well as a strong-type reflected shock with $\vec v_3$ parallel to the wall.
    Choose an opposite wall whose extension to a line 
    passes through the point $\vec x/t=\vec v_2$ (so that the slip condition
    \myeqref{eq:slipcond}
    in the $2$-sector is satisfied).
    If the angle between the two walls is chosen small enough, then the angle condition
    is violated.
    We can choose this local RR transonic as follows: by Proposition \myref{prop:shockpolar} applied to
    the incident shock polar, for sufficiently large $\tau$,
    $M_2\downarrow 1$, so $\tau_*\downarrow 0$ for the corresponding reflected shock polar 
    (Proposition \myref{prop:shockpolar}).
    For $\tau\approx\tau_*$, $M_3<1$ which necessarily happens as $\tau$ grows.
\end{proof}

\begin{corollary}
    \label{cor:strongunstable}%
    Strong-type trivial RR is not always structurally stable.
\end{corollary}
\begin{proof}
    Given $M_1,\alpha,\theta$ for a strong-type trivial RR,
    perturb to $\alpha-\delta$, $\theta+\delta$
    for some small $\delta>0$.
    Then we are in the situation of Theorem \myref{th:detach-wrong} where global RR cannot exist.
\end{proof}

In summary, we have obtained two separate results. First, the detachment criterion is not universally correct. 
Note however that we have discussed only some cases with $\theta>90^\circ$. It would be interesting to find extensions to the 
classical case $\theta<90^\circ=\alpha$; in that case, the detachment criterion is probably correct. 
A new RR$\rightarrow$MR transition criterion is proposed in Section \myref{section:newcrit}. 

Second, while weak-type transonic trivial RR is structurally stable, \emph{strong}-type is \emph{not} (Corollary \ref{cor:strongunstable}). 
This provides an important new answer to the weak-strong problem. Note that historically, \emph{dynamic} stability, 
i.e.\ under perturbation of the \emph{initial} data, has been considered. \cite{elling-liu-rims05} observes 
numerically that both weak- and strong-type reflection are dynamically stable, so any mathematical 
result to the contrary appears to use an overly restrictive definition of stability.

\section{Potential flow}
\mylabel{section:potf}

Self-similar potential flow is the second-order quasilinear PDE
\begin{alignat}{1}
    \nabla\cdot(\rho\nabla\chi)+2\rho &= 0. \myeqlabel{eq:sspf-divform}
\end{alignat}
Here 
\begin{alignat}{1}
    \chi &= \psi-\frac12|\vec\xi|^2.
\end{alignat}
$\vec\xi=(\xi,\eta)=\vec x/t$ are \defm{similarity coordinates}.
$\chi$ is called \defm{pseudo-potential}. $\psi$ is the \defm{velocity potential}: physical velocity is
$$\vec v=\nabla\psi.$$
Moreover, density is
\begin{alignat}{1}
    \rho &= \pi^{-1}(-\chi-\frac12|\nabla\chi|^2).\myeqlabel{eq:rhofun}
\end{alignat}
$\pi$ satisfies
\begin{alignat}{1}
  \frac{d\pi}{d\rho}=\rho^{-1}\frac{dp}{d\rho}=\rho^{-1}c^2,\qquad(\pi^{-1})'=\rho c^{-2} 
  \myeqlabel{eq:pideriv}
\end{alignat}
where
$$p(\rho)=\frac{\rho_0c_0^2}{\gamma}\left(\frac{\rho}{\rho_0}\right)^\gamma$$
is the equation of state ($\rho_0,c_0$ free parameters). 
The \defm{ratio of heats}\footnote{also: \defm{isentropic coefficient}}
is restricted to $\gamma\in(1,\infty)$. 
Differentiation of \myeqref{eq:sspf-divform} yields the non-divergence form
\begin{alignat}{1}
  (c^2I-(\nabla\psi-\vec\xi)^2):\nabla^2\psi &= 0. \myeqlabel{eq:nondivpsi}
\end{alignat}
Here $A:B$ is the Frobenius product $\trace(A^TB)$, $\vec w^2:=\vec w\otimes\vec w=\vec w\vec w^T$ (as opposed
to $|\vec w|^2=\vec w\cdot\vec w$) and $\nabla^2$ is accordingly the Hessian.
In coordinates:
\begin{alignat}{1}
  (c^2-(\psi_\xi-\xi)^2)\psi_{\xi\xi}-2(\psi_\xi-\xi)(\psi_\eta-\eta)\psi_{\xi\eta}+(c^2-(\psi_\eta-\eta)^2)\psi_{\eta\eta} &= 0.
\end{alignat}
$c$ is the sound speed, defined by
\begin{alignat}{1}
    c^2 &= c_0^2+(1-\gamma)(\chi+\frac12|\nabla\chi|^2).
\end{alignat}
It is sometimes more convenient to use the form
\begin{alignat}{1}
    (c^2I-\nabla\chi^2):\nabla^2\chi + 2c^2 - |\nabla\chi|^2 = 0. \myeqlabel{eq:nondivchi}
\end{alignat}
This form is manifestly translation-invariant. Translation is nontrivial: in $(t,x,y)$ coordinates it corresponds
to a change of inertial frame
\begin{alignat}{1}
  \vec v &\leftarrow\vec v-\vec w,\qquad \vec\xi=\vec x/t \leftarrow\vec\xi-\vec w, \myeqlabel{eq:chframe}
\end{alignat}
where $\vec w$ is the velocity of the new frame relative to the old one. Obviously the \defm{pseudo-velocity}
$$\vec z:=\nabla\chi=\nabla\psi-\vec\xi$$
does not change.

Self-similar potential flow is mixed-type; the local type is determined by 
the coefficient matrix $c^2I-\nabla\chi^2$ which is positive definite if and 
only if $L<1$, where 
\begin{alignat}{1}
    L &:= \frac{|\vec z|}{c}=\frac{|\vec v-\vec x/t|}{c} \mylabel{eq:pseudo-Mach}
\end{alignat}
is called \defm{pseudo-Mach number}. 
For $L>1$ the equation is hyperbolic; parabolic is $L=1$.
$L$ and $\vec z$ are the Mach number and velocity perceived by an observer traveling on the
ray $\vec x=t\vec\xi$.

On a solid wall the \defm{slip condition}
\begin{alignat}{1}
  \nabla\chi\cdot\vec n &= 0 \myeqlabel{eq:slipcond}
\end{alignat}
holds; for an observer traveling on the wall it corresponds to the usual 
\begin{alignat}{1}
  \vec v\cdot\vec n &= \nabla\psi\cdot\vec n =  0 \myeqlabel{eq:psislip}
\end{alignat}

\section{Shock conditions}

The weak solutions of potential flow are defined by \myeqref{eq:sspf-divform}.
The corresponding Rankine-Hugoniot condition is 
\begin{alignat}{1}
  \rho_uz^n_u &= \rho_dz^n_d \myeqlabel{eq:rh-z} 
\end{alignat}
where $u,d$ indicate the limits on the \defm{upstream} and \defm{downstream} side and $z^n$, 
$z^t$ are the normal and tangential
component of $\vec z$.
As the equation is second-order, we must additionally require continuity of the potential:
\begin{alignat}{1}
  \psi_u &= \psi_d.
\end{alignat}
By taking a tangential derivative, we obtain
\begin{alignat}{1}
  z^t_u &= z^t_d =: z^t.
\end{alignat}

Observing that $\sigma=\vec\xi\cdot\vec n$ is the shock speed, we obtain the more familiar form
\begin{alignat}{1}
  \rho_uv^n_u - \rho_dv^n_d &= \sigma(\rho_u-\rho_d), \myeqlabel{eq:rh-v} \\
  v^t_u &= v^t_d =: v^t. \myeqlabel{eq:vtan}
\end{alignat}

Fix the unit shock normal $\vec n$ so that $z^n_u>0$ which implies $z^n_d>0$ as well.
To avoid expansion shocks we must require the admissibility condition $z^n_u\geq z^n_d$, which is equivalent to
\begin{alignat}{1}
  v^n_u &\geq v^n_d.
\end{alignat}
We choose the unit tangent $\vec t$ to be $90^\circ$ counterclockwise from $\vec n$.

By \myeqref{eq:vtan} the tangential components of the velocity are continuous across the shock,
so the velocity jump is normal. 
Assuming $v^n_u>v^n_d$ (positive shock strength), we can express the shock normal as 
\begin{alignat}{1}
  \vec n &= \frac{\vec v_u-\vec v_d}{|\vec v_u-\vec v_d|}. \myeqlabel{eq:normal-v}
\end{alignat}

\section{Nonexistence of some global RR}

We start with some facts about the shock polar.
\begin{proposition}
  \mylabel{prop:shockpolar}%
    Consider arbitrary $c_u,\rho_u>0$ and $M_u\in(1,\infty)$ and set $\vec v_u=(M_uc_u,0)$.
    For each $\beta\in(-90^\circ,90^\circ)$ there is a steady shock 
    with downstream unit normal $\vec n=(\cos\beta,\sin\beta)$.
    Its downstream data depends smoothly on $\beta$. Let $\tau$ be counterclockwise angle from
    $\vec v_u$ to $\vec v_d$.
    We restrict 
    $|\beta|<\arccos\frac{1}{M_u}$
    so that the shock is admissible.

    Then the \defm{shock polar} $\beta\mapsto\vec v_d$ is smooth and strictly convex, 
    with $\partial_\beta\vec v_d$ 
    nowhere zero. 

    There is an angle $\tau_*\in(0^\circ,90^\circ)$ so that each $\tau\in(-\tau_*,\tau_*)$
    is attained for two different $\beta$.
    The one with smaller $|\vec v_d|$ yields a strong-type shock, the other one weak-type.
    For $|\tau|=\tau_*$ they are identical and critical-type.

    There is a $\tau_s\in(0,\tau_*)$ so that the weak-type shocks are supersonic
    for $|\tau|>\tau_s$, transonic for $|\tau|<\tau_s$. The other types are always transonic.

    If $M_u\downarrow 1$ with $\rho_u,c_u$ fixed, then $\tau_*\downarrow 0$.
\end{proposition}
\begin{proof}
    Most has been shown in 
    \cite[Theorem 1]{elling-sonic-potf} 
    and \cite[Proposition 2.10]{elling-liu-pmeyer}; we only need to prove the last statement.
    Admissible shocks are those for $|\beta|\leq\arccos\frac{1}{M_u}$. As $M_u\downarrow 1$,
    this range shrinks to $\{0\}$. 
    By continuity, all points on the shock polar approach $\vec v_u$.
    In particular $\tau_*\downarrow 0$.
\end{proof}

\begin{figure}
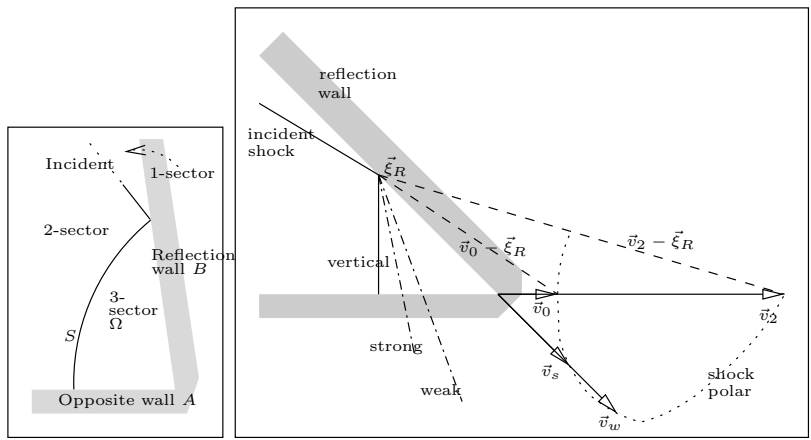

\input{trans.pstex_t}
\input{polarg2.pstex_t}
\caption{Left: transonic global RR, angle condition satisfied. 
Self-similar potential flow is elliptic in
$\Omega$, hyperbolic elsewhere. Right: Shock polar argument.}
\mylabel{fig:trans-left}
\mylabel{fig:polarg-right}
\end{figure}

Let $\Omega$ be the $3$-sector excluding boundary (Figure \ref{fig:trans-left} left), $A$ opposite wall, $B$ reflection wall, $S$ reflected shock,
each not containing its endpoints.

\begin{proposition}
    \mylabel{prop:vvert}%
    Consider the setting of Theorem \myref{th:detach-wrong},
    with $(0,0)$ the wall-wall corner (see Figure \myref{fig:polarg-right}).
    The vertical straight 
    shock with upstream data $\vec v_2,\rho_2,c_2$ through the reflection point $\vec\xi_R$
    has a downstream velocity $\vec v_0=(v^x_0,0)$ with $v^x_0>0$. 
    The same holds for all vertical shocks to the right of it.
\end{proposition}
\begin{proof}
    First change to a coordinate frame with origin in the reflection point. 
    For this observer the $2$-sector velocity is $\vec v_2-\vec\xi_R$, which points into the reflection wall (Figure \ref{fig:polarg-right} right).
    Consider shocks with that upstream velocity and upstream density $\rho_2$ and sound speed $c_2$;
    let $\vec v_d-\vec\xi_R$ be the downstream velocity.
    Let $\beta$ be the counterclockwise angle from $\vec v_2-\vec\xi_R$ to shock downstream normal $\vec n$;
    by \myeqref{eq:normal-v} $\vec n$ is a positive multiple of $\vec v_2-\vec v_d$. 

    By design (slip condition), 
    the velocities $\vec v_w-\vec\xi_R$ and $\vec v_s-\vec\xi_R$ for weak-type and strong-type reflected shock
    are on the extension of the reflection wall into a line
    (see Figure \myref{fig:polarg-right} right). 
    By assumption of Theorem \myref{th:detach-wrong} the angle condition is violated,
    so the strong-type and therefore the weak-type reflected shock tangent in
    the reflection point are down and strictly right. 
    Thus the vertical shock through the reflection point has smaller $|\beta|$ than either type, so by 
    strict convexity of
    the shock polar (Proposition \myref{prop:shockpolar}) $\vec v_0-\vec\xi_R$ points into the reflection wall
    (see Figure \myref{fig:polarg-right}). Therefore $v^x_0>0$, since $\vec v_2$ and thus $\vec v_0$ are horizontal.

    By \cite[Proposition 2.9]{elling-liu-pmeyer}, vertical shocks more to the right have $v^x_d-\xi_R>v^x_0-\xi_R$ (because they
    are weaker, so $\vec v_d$ is closer to $\vec v_2$). Hence $v^x_d>v^x_0>0$ as well.
\end{proof}

\begin{proof}[Proof of Theorem \myref{th:detach-wrong}]
    Consider the same coordinates as in the statement of Proposition \myref{prop:vvert}.
    Restrict $\psi$ to $\overline\Omega$ (taking its $\Omega$-side limits on $\partial\Omega$).
    Let $\psi_0$ be the value of $\psi$ in the reflection point $\vec\xi_R$.
    Let $S_0$ be the straight vertical shock through the reflection point;
    let $\sigma_0$ be its $\xi$ coordinate.

    Consider a transonic global RR.

    Again by assumption the angle condition does not hold,
    so the reflection point shock tangent 
    points down and strictly right (as in Figure \myref{fig:nopert-right} right, as opposed to
    Figure \myref{fig:trans-left} left).
    The upstream velocity $\nabla\psi=\vec v_2$ has $\psi_x=v^x_2>0$, so necessarily
    $\psi>\psi_0$ at the shock \emph{near} the reflection point $\vec\xi_R$. 
    Therefore, the global maximum of $\psi$ over $\overline\Omega$ (which must be attained
    since $\overline\Omega$ is compact and $\psi$ continuous)
    is greater than $\psi_0$ and not attained in $\vec\xi_R$. 
    
    Consider a maximum $>\psi_0$ in a point $\vec\xi\in\overline S-\{\vec\xi_R\}$.
    The shock tangent is vertical in $\vec\xi$
    (by $\psi_t=0$ for a maximum at $S$; by the 
    slip condition \myeqref{eq:psislip} at $A$ for a maximum in the point
    where $\overline S$ meets $A$).
    Moreover, $\psi>\psi_0$ implies the shock is right
    of the vertical reflection point shock because $\psi$ is continuous across the shock and
    $\psi_x=v^x_2>0$ on the upstream side. 
    Hence by Proposition \myref{prop:vvert},
    $\psi_\xi=v^x\geq v^x_0>0$.
    This is incompatible with a local maximum.

    Hence $\psi$ does not attain its $\overline\Omega$-maximum anywhere on $\overline S$.
    Then the same is true for 
    $$\hat\psi:=\psi+\delta\xi$$
    if we choose $\delta>0$ sufficiently small. 
    By linearity   
    $$(I-c^{-2}\nabla\chi^2):\nabla^2\hat\psi=(I-c^{-2}\nabla\chi^2):\nabla^2\psi=0,$$
    so by the strong maximum principle the maximum is not attained in $\overline\Omega$ either
    (if we choose $\delta>0$ so small that $\hat\psi$, like $\psi$, cannot be constant).
    Moreover
    $$\nabla\hat\psi\cdot\vec n=\nabla\psi\cdot\vec n+\delta n^x=\delta n^x\geq 0$$
    on $A$ and $B$, so the Hopf lemma rules out local maxima there. 
    
    Finally, the boundary conditions on $A$ and $B$, combined with $C^1$ continuity in $0$ (Definition \ref{def:flow-class}),
    imply $\nabla\psi(0)=0$, so $\hat\psi_\xi(0)=\psi_\xi(0)+\delta>0$, thus a local maximum in $0$ is
    impossible. 
    
    We have ruled out every possible global maximum point in $\overline\Omega$. The contradiction
    demonstrates that no $\psi$ with the desired properties exists. 
\end{proof}

\section{Numerical comparison}

\mylabel{section:numerics}

\begin{figure}
\includegraphics[width=.49\linewidth]{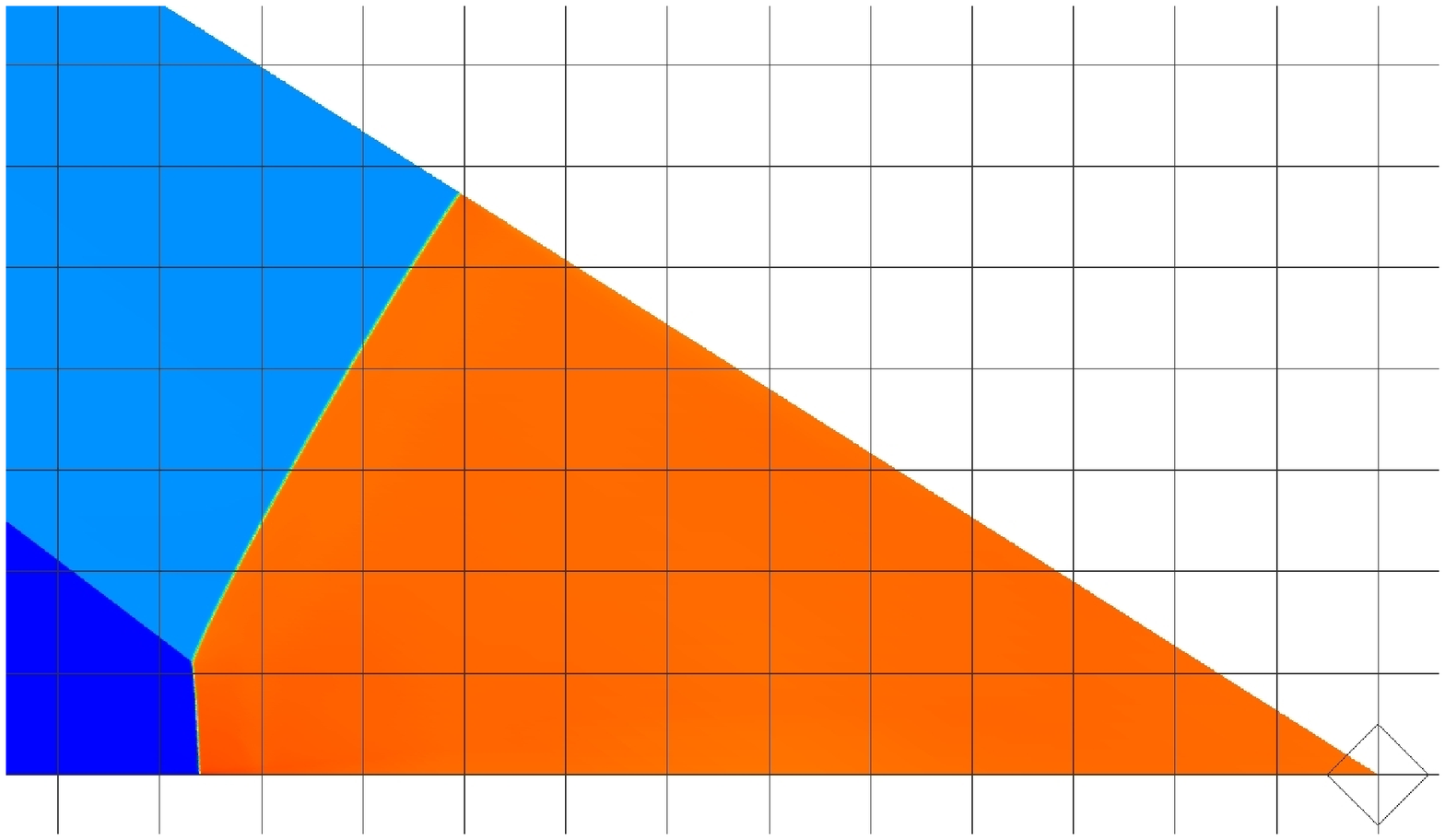}
\includegraphics[width=.49\linewidth]{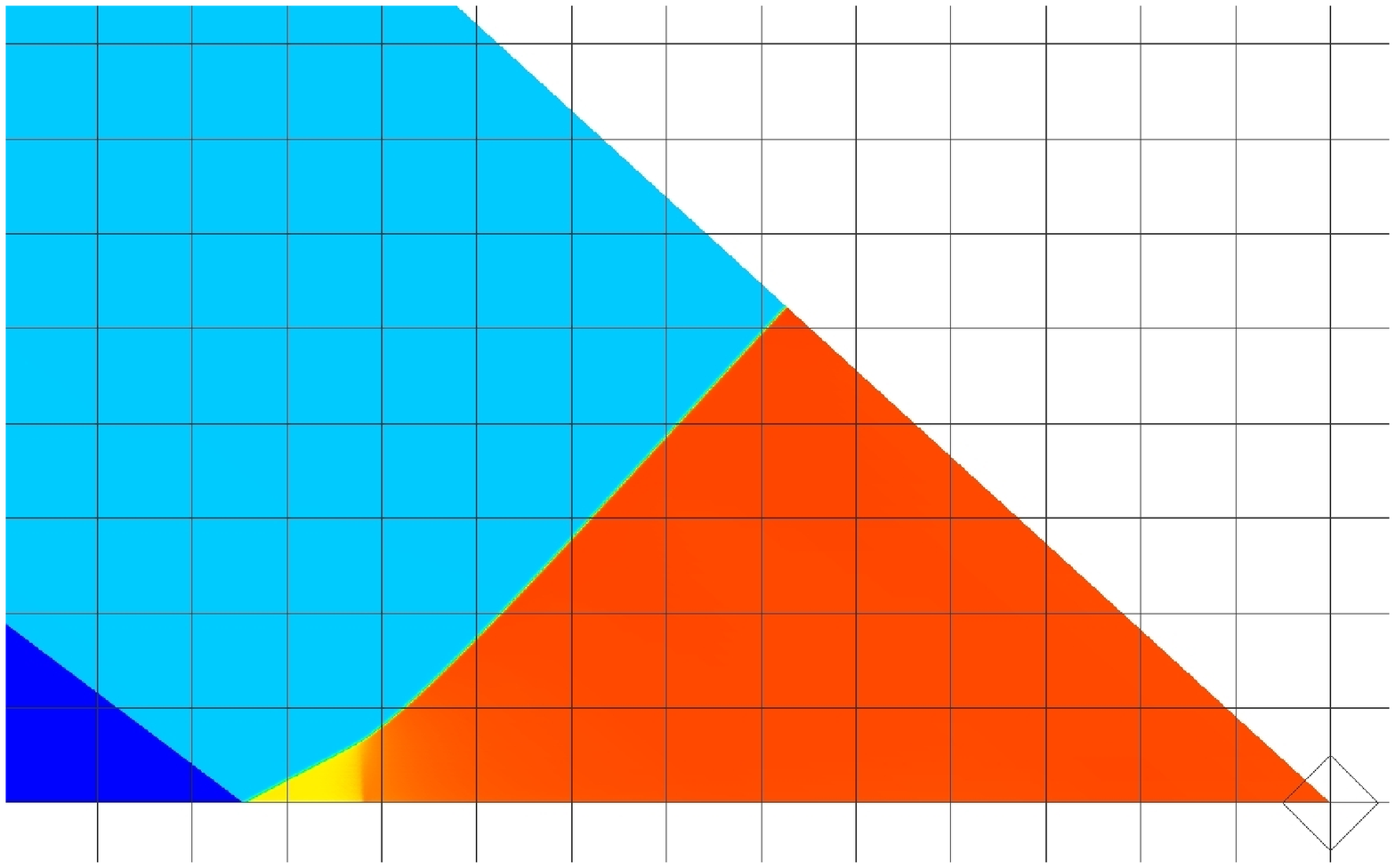}
\caption{[Reflection and opposite wall are exchanged in these diagrams.]
Left: MR for $\theta=147.9^\circ$. Instead of meeting on the bottom wall,
incident (left) and reflected (right) shock meet in a triple point
with a vertical shock (Mach stem).
Right: RR for $\theta=137.9^\circ$.
Near the reflection point the flow is hyperbolic; the transition
to elliptic is discontinuous (MR).}
\mylabel{fig:highertheta-left}%
\mylabel{fig:lowertheta-right}%
\end{figure}

Theorem \myref{th:detach-wrong} concerns the range of parameters with transonic 
weak-type RR, which is so narrow (see Figure \myref{fig:transition} right) that 
numerics and experiments have not been able to settle questions for these flows. 
However, the range with \emph{supersonic} weak-type RR violating the angle 
condition is much larger and certainly interesting by itself. 

For $\gamma=7/5$, $M_1\approx3$ and $\alpha=0^\circ$, $\theta=142.9^\circ$
corresponds exactly to a strong-type trivial RR (i.e.\ strong-type shock perpendicular
to opposite wall).
We change $\theta$ by $5^\circ$ to $147.9^\circ$ without changing $\theta+\alpha$ or $M_I$.
This way the opposite wall angle changes, but not the local RR parameters.
The numerical results in Figure \myref{fig:highertheta-left} left show an MR.

We also study the opposite perturbation, to $\theta=137.9^\circ$ (see Figure
\myref{fig:lowertheta-right} right). As expected there
is still a local RR. The shock is essentially the strong-type shock, 
except in a small neighbourhood of the reflection point where it is weak-type and slightly hyperbolic.
As $\theta\uparrow142.9^\circ$, this neighbourhood shrinks to zero;
it appears that the pattern converges to the trivial strong-type RR in this manner.
This is why strong-type reflections are observed at a large scale sometimes.
Note that there is a MR as well: at the transition from hyperbolic to elliptic.

The calculations were made with a second-order scheme on an unstructured grid; 
other choices have no influence on the qualitative structure (RR vs.\ MR).

In principle Definition \myref{def:flow-class} and Theorem \myref{th:detach-wrong} 
could be extended to the supersonic cases.
But while in the transonic case the flow is simple and predictable, at least for
small perturbations from trivial RR, the supersonic cases can have several different 
qualitative structures. \cite{chen-feldman-selfsim-journal,elling-rrefl} construct
self-similar RR with a continuous transition from hyperbolic to elliptic in the $3$-sector,
but Figure \myref{fig:lowertheta-right} right shows a MR, i.e.\ a discontinuous transition;
double Mach reflection and other more complicated flows are possible too. Proving
nonexistence in function classes large enough to accomodate all these structures is 
far beyond present-day techniques.

\section{Interpretation}

\mylabel{section:newcrit}

\begin{figure}
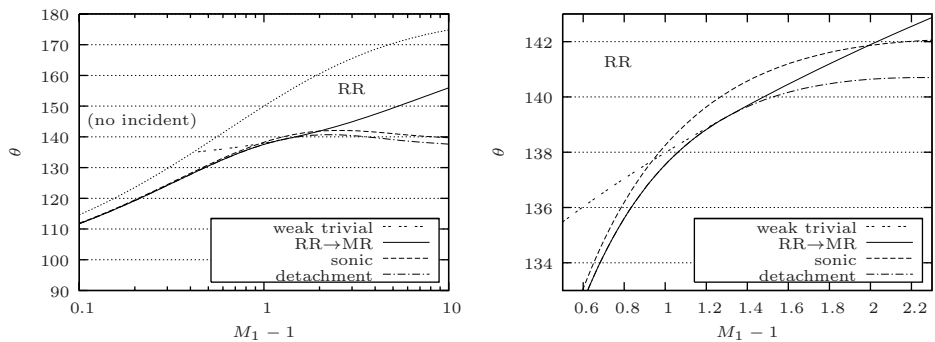

\input{rr-hor-potf.pstex_t}
\input{rr-hor-detail-potf.pstex_t}
\caption{Proposed RR$\rightarrow$MR transition, for $\gamma=7/5$ potential flow
with $\alpha=0^\circ$. Right: detail.
Weak-type reflection is supersonic above the ``sonic'' curve, transonic below;
neither type exists below the ``detach'' curve. Theorem \myref{th:detach-wrong}
rules out global RR below the solid curve.}
\mylabel{fig:transition}
\end{figure}

Despite the theorem and numerical examples, it is likely that the detachment criterion is still valid
over a large part of the parameter space. In particular, the author believes that it is correct
in the classical case $\alpha=90^\circ$, $\theta<90^\circ$. 
We propose the following new criterion:
\myquote{%
    The global flow is RR if and only if local RR exists and 
    angle condition is satisfied.
}
Strong-type RR would appear only in the trivial right-angle borderline case separating global RR and global MR.

Figure \myref{fig:transition} shows the regions predicted by this criterion for $\gamma=7/5$ and 
$\alpha=0^\circ$. 
Cases that have already been treated by construction of an exact solution or another rigorous method:
\begin{enumerate}
\item Nonexistence of global RR below the detachment criterion is trivial.
\item Nonexistence of global RR below the solid curve is done in this article
    for transonic weak-type RR by Theorem \myref{th:detach-wrong}.
\item Existence of global transonic RR is done by \cite{elling-sonic-potf} for some neighbourhood
    of each point on the transonic part of the ``weak trivial'' curve in Figure \myref{fig:transition}
    (excluding endpoints).
\item 
  \cite{chen-feldman-selfsim-journal,elling-rrefl} construct global supersonic RR for 
  some of the supersonic parameters, in particular some neighbourhood of 
  each point of the supersonic 
  part of the ``weak trivial'' curve in Figure \myref{fig:transition}
  (excluding endpoints).
\end{enumerate}
In principle, \cite{chen-feldman-selfsim-journal,elling-liu-pmeyer,elling-rrefl,elling-sonic-potf} go
a long way towards constructing global RR in all cases not covered so far.
In comparison, global MR is very difficult:
the triple point, well-known to be theoretically
impossible (\defm{von Neumann paradox}), has a very complicated detail structure,
according to numerical results of Hunter/Tesdall (\cite{hunter-tesdall,tesdall-sanders-keyfitz}), see also 
\cite{vasiliev-kraiko,skews-ashworth}),

\bibliographystyle{plain}
\bibliography{../../../pmeyer/elling}

\end{document}